\documentclass[journal]{IEEEtran}
\usepackage[utf8]{inputenc}

\usepackage{times}
\usepackage{amssymb}  
\usepackage{cite}     
\usepackage{textgreek}
\usepackage{upref}
\usepackage{amsfonts}
\usepackage{verbatim}
\usepackage[dvipsnames,usenames]{color}
\usepackage{cite}
\usepackage[inline]{enumitem}
\usepackage{caption}
\usepackage{balance}
\usepackage{nicefrac}
\usepackage{amsthm}
\usepackage{algorithm}
\usepackage{times,color}
\usepackage{enumitem}
\usepackage{float}
\usepackage{color}
\usepackage{mathtools}
\usepackage{graphicx}
\usepackage{listings}
\usepackage{tikz}
\usepackage{multicol,multirow}
\usepackage{psfrag}
\usepackage{bbm}
\usepackage{verbatim}
\usepackage{soul}
\usepackage{varwidth}

\usepackage{pgfplots}
\usepackage{hhline}
\usepackage{xcolor,colortbl}
\usepackage{algorithm}
\usepackage[noend]{algpseudocode}

\usepackage{enumitem}

\usetikzlibrary{%
	arrows,%
	automata,%
	calc,%
	patterns,%
	plotmarks,%
	positioning,%
	shapes.geometric,%
	shapes,%
	snakes,%
	decorations.pathmorphing,%
	decorations.shapes,%
	graphs,%
	chains,%
	arrows.meta, %
}

\algnewcommand{\Initialize}[1]{%
	\State \textbf{Initialize:}
	\Statex \hspace*{\algorithmicindent}\parbox[t]{.8\linewidth}{\raggedright #1}
}

\newcommand{\del}{\delta}
\newcommand{\ins}{\iota}

\newcommand{\be}[1]{\begin{equation}\label{#1}}
\newcommand{\ee}{\end{equation}}

\newcommand{\bc}{\begin{center}}
\newcommand{\ec}{\end{center}}

\newcommand{\cC}{{\cal C}}
\newcommand{\cD}{{\cal D}}

\newcommand{\cI}{{\cal I}}

\newcommand{\cO}{{\cal O}}

\newcommand{\x}{\mathbf{x}}
\newcommand{\y}{\mathbf{y}}

\renewcommand{\leq}{\leqslant}

\renewcommand{\geq}{\geqslant}

\newcommand{\Cref}[1]{Co\-rol\-la\-ry\,\ref{#1}}

\newtheorem{theorem}{Theorem}

\newtheorem{definition}{Definition}

\newtheorem{remark}{Remark}

\newtheorem{construction}{Construction}

\definecolor{Codecolor}{named}{White}  

\newcommand{\Copen}{\mbox{\{\kern-5.50pt\{}}
\newcommand{\Cclose}{\mbox{\}\kern-5.50pt\}}}
\newcommand{\Cslash}{\mbox{$\backslash\kern-6.02pt\backslash$}}

\newcommand{\sepb}{{\color{green!50!black}0}}
\newcommand{\sepr}{{\color{red}1}}

\def \th{^\text{th}}

\newif\iflong
\longtrue

\begin{document}

\title{Detecting Deletions and Insertions in Concatenated Strings with Optimal Redundancy}%
\author{%
\IEEEauthorblockN{Serge Kas Hanna and  Rawad Bitar \\ Institute for Communications Engineering, Technical University of Munich, Germany\\ Emails: \{serge.k.hanna, rawad.bitar\}@tum.de
}
	\thanks{This project has received funding from the European Research Council (ERC) under the European Union's Horizon 2020 research and innovation programme (grant agreement No. 801434) and from the Technical University of Munich - Institute for Advanced Studies, funded by the German Excellence Initiative
and European Union Seventh Framework Programme under Grant Agreement
No. 291763.}\vspace{-4ex}
}

\maketitle

\begin{abstract}
        We study codes that can detect the exact number of deletions and insertions in concatenated binary strings. We construct optimal codes for the case of deletions. We prove the optimality of these codes by deriving a converse result which shows that the redundancy of our codes is asymptotically optimal in the number of deletions among all families of deletion detecting codes, and particularly optimal among block-by-block decodable codes. For the case of insertions, we construct codes that can detect up to $2$ insertions in each concatenated binary string.
\end{abstract}

\section{Introduction}\label{sec:intro}

We consider the problem of detecting deletions and insertions in concatenated binary strings. More precisely, we consider strings of the form \mbox{$\x=\langle \x^1, \x^2, \cdots, \x^{n/\ell} \rangle \in \mathbb{F}_2^n$}, that are implicitly divided into $n/\ell$ disjoint substrings  $\x^1, \x^2, \cdots, \x^{n/\ell}$, referred to as blocks, each being of size~$\ell$. Suppose that $\x$ is affected by worst-case deletions and insertions resulting in a string $\y$. We are interested in constructing codes that can detect the exact number of deletions and insertions that have occurred in each of the blocks $\x^1, \x^2, \cdots, \x^{n/\ell}$, based on the observed string $\y$. Such codes can consequently determine the boundaries of each block in~$\y$. Furthermore, these codes enable localizing deletions and insertions and identifying their presence in certain parts of an information string. Thus, they have several potential applications such as coding for segments edit channels~\cite{A18,L10}, marker-based constructions for coded trace reconstruction~\cite{J20}, file synchronization~\cite{R15}; in addition to detecting mutations in parts of a DNA sequence, and retaining synchronization in sequential transmission.

Most of the literature has focused on the correction of deletions and insertions, under the assumption that the codeword boundaries are known at the decoder. Levenshtein~\cite{L66}  derived fundamental limits which show that the optimal number of redundant bits needed to correct $\delta$ worst-case deletions in a binary string of length $n$ is $\Theta(\del\log(n/\del))$. 
Levenshtein also showed that the code constructed by Varshamov and Tenengolts (VT codes)~\cite{VT65} is capable of correcting a single deletion and has asymptotically optimal redundancy in~$n$. There have been lots of works in the past few years on the classical problem of constructing codes that correct deletions and insertions~\cite{B16,GC,H19,Chen18,SimaIT,SimaSYS}. The state-of-the-art results in \cite{Chen18,SimaIT,SimaSYS} give codes with $\mathcal{O}(\delta \log n)$ redundancy. Some of the recent works also studied the case where the deletions and insertions occur in bursts or are localized within certain parts of the codeword~\cite{Sch17,GClocalized,A20}. 

As previously mentioned, the aforementioned works in the literature have a common requirement: the codeword boundaries must be known at the decoder in order to successfully correct the errors. In fact, one can easily show that if multiple codewords of the single deletion correcting VT code are concatenated and transmitted over a channel that deletes at most one bit in each codeword, then the decoder cannot determine the boundaries of these codewords with certainty. In other words, the decoder cannot detect the number of deletions that have occurred in each codeword. Therefore, the problem that we study in this paper cannot be solved by concatenating codewords that belong to classical deletion and insertion correcting codes.

Codes for detecting deletions and insertions have been previously studied in~\cite{K03,P13} under a different definition than the one we use in this paper. The definition used in~\cite{P13} for a deletion detecting code is as follows. A code $\cC\subseteq \mathbb{F}_2^n$ is said to be a $\del$-deletion detecting code in~\cite{P13}, if for any $\x\in \cC$, the process of deleting any $\del'\leq \del$ bits from $\x$ and then appending arbitrary $\del'$ bits at the end, does not produce a codeword in~$\cC$. The authors in~\cite{K03} consider a similar definition and focus on non-binary codes. The main difference between our definition and the definitions in~\cite{K03,P13} is that we require the decoder to detect the exact number of deletions as opposed to only identifying whether deletions have occurred in a block or not. Namely, our (informal) definition is the following. We say that a code $\cC$ detects up to $\del$ deletions in each of the blocks of \mbox{$\x=\langle \x^1, \x^2, \cdots, \x^{n/\ell} \rangle \in \cC$}, if and only if there exists a decoder that can determine the exact number of deletions that have occurred in each block after $\x$ is affected by worst-case deletions. The difference between the two previous definitions is crucial in the setting of concatenated strings which we consider in this paper. Namely, in the presence of multiple concatenated blocks, detecting the exact number of deletions in a given block allows the decoder to determine the boundary of that block, and consequently proceed to decode the next block.

Our main contributions and the organization of this paper are summarized as follows. A formal definition of the problem is given in Section~\ref{sec:problem}. In Section~\ref{sec:detecting_del}, we construct an explicit code that detects up to $\del$ deletions in each block of a codeword $\x$. The code is encodable and decodable in linear time $\cO(n)$, and its redundancy is $(2\del+1)(n/\ell-1)$ bits, where $\ell$ is the length of each block in the codeword. Then, we derive lower bounds on the redundancy of codes that detect deletions. These bounds show that the redundancy of our codes is optimal among all block-by-block decodable codes (Definition~\ref{def3}), and asymptotically optimal in $\delta$ among all codes that detect up to $\del$ deletions. In Section~\ref{sec:detecting_ins}, we present two code constructions that detect up to $1$ and up to $2$ insertions per block. We conclude the paper with some open problems in Section~\ref{sec:conc}.

\section{Problem Statement and Notation}\label{sec:problem}
We start by introducing some of the notations used throughout the paper.  Let $[n]$ be the set of integers from $1$ to $n$~(inclusive), and let $[i,j]$ be the set of integers from $i$ to $j$ (inclusive). Let $\mathbf{1}^i$ and $\mathbf{0}^j$ denote strings of $i$ consecutive ones and $j$ consecutive zeros, respectively. For a string \mbox{$\mathbf{x}=(x_1,x_2,\ldots,x_n)\in \mathbb{F}_2^n$}, we use $x_i$, \mbox{$i=1,2,\ldots,n$}, to refer to the $i\th$ bit of $\mathbf{x}$. We write $\mathbf{x}_{[i,j]}=(x_i,x_{i+1},\ldots,x_j)$ as the substring of $\mathbf{x}$ which contains the consecutive bits ranging from index $i$ to index $j$. 

\begin{definition}
\label{def1}
For a given $\ell \in \mathbb{Z}^+$ with $\ell<n$, the $j\th$ block of $\mathbf{x}\in \mathbb{F}_2^n$ is defined as the substring $$\mathbf{x}^j\triangleq \mathbf{x}_{[1+(j-1)\ell, j\ell]}=(x_{1+(j-1)\ell},x_{2+(j-1)\ell},\ldots,x_{j\ell}),$$
where $j \in [n/\ell]$.
\end{definition}

In Definition~\ref{def1}, and throughout the paper, we assume that $\ell$ divides~$n$. We use $\langle \mathbf{a},\mathbf{b}\rangle$ to refer to the concatenation of two strings $\mathbf{a}$ and~$\mathbf{b}$. In this paper, we study the problem of constructing codes for detecting deletions and insertions in concatenated strings. We focus on the case where these strings are binary.
\iflong
All logarithms in this paper are of base $2$.
\else
\fi
\begin{definition}
\label{def2}
Let $\mathcal{C}_{\tau}(\ell,n)\subseteq \mathbb{F}_2^n$, with $\tau< \ell \leq n/2$, be a code of length $n$ that contains codewords of the form $$\mathbf{x}=\langle \mathbf{x}^1,\mathbf{x}^2,\ldots,\mathbf{x}^{n/\ell}\rangle,$$ where $\mathbf{x}^j \in \mathbb{F}_2^{\ell}$ for all $j\in [n/\ell]$. Suppose that $\mathbf{x}\in \mathcal{C}_{\tau}(\ell,n)$ is affected by at most $\tau=\del+\ins$ errors in each block, resulting in a string \mbox{$\mathbf{y}\in \mathbb{F}_2^*$}, where $\del$ and $\ins$ denote the number of deletions and insertions, respectively. The code $\mathcal{C}_{\tau}(\ell,n)$ is said to detect up to $\tau$ errors per block, if and only if there exists a decoding function $$\mathrm{Dec}(\mathbf{y}) : \mathbb{F}_2^* \to \mathbb{Z}_{\tau+1}^{n/\ell} \times \mathbb{Z}_{\tau+1}^{n/\ell},$$ that outputs the exact numbers of deletions and insertions that have occurred in each block $\mathbf{x}^j$, for all $\mathbf{x}$ and $\mathbf{y}$.
\end{definition}

For example, a code $\cC_1(\ell,n)$ is said to detect up to $\tau=1$ error per block, if and only if there exists a decoding function $\mathrm{Dec}(\mathbf{y}) : \mathbb{F}_2^* \to \mathbb{Z}_{2}^{n/\ell} \times \mathbb{Z}_{2}^{n/\ell}$ whose output for a given block $j\in [n/\ell]$ is: \begin{enumerate*}[label={\textit{(\roman*)}}] \item $(0,0)$ if $\x^j$ was not affected by any error; \item $(1,0)$ if $\x^j$ was affected by exactly $1$ deletion; \item $(0,1)$ if $\x^j$ was affected by exactly $1$ insertion. \end{enumerate*} Note that the decoding requirement in Definition~\ref{def2} of detecting the {\em exact} numbers of deletions and insertions per block, is equivalent to detecting the boundaries of each block in the observed string~$\mathbf{y}$. Furthermore, we use $\cC_{\del}(\ell,n)$ to refer to codes that can only detect up to $\del$ deletions per block, i.e., $\tau=\del$ and $\ins=0$. Similarly, we use $\cC_{\ins}(\ell,n)$ to refer to codes that can only detect up to $\ins$ insertions per block.

\begin{definition}
\label{def3}
Let $\cC(\ell,n)$ be a code that follows Definition~\ref{def2}. Consider a codeword $\x \in \cC(\ell,n)$ that is affected by at most $\tau$ errors in each of its blocks, resulting in $\y$. Let $\alpha_j$ be the starting position of block $j$ in $\y$, with $\alpha_1=1$. The code $\cC(\ell,n)$ is said to be block-by-block decodable, if and only if there exists a decoder that can output the exact numbers of deletions and insertions in any block $\x^j$ by only processing the bits in~$\y_{[\alpha_j,\alpha_j+\ell'-1]}$, where $\ell'$ is the maximum length of the block in $\y$ given by: \begin{enumerate*}[label={\textit{(\roman*)}}] \item $\ell'=\ell+\tau$ for a code $\mathcal{C}_{\tau}(\ell,n)$ that detects up to $\tau$ errors; \item $\ell'=\ell$ for a code $\cC_{\del}(\ell,n)$ that detects up to $\del$ deletions; \item $\ell'=\ell+\ins$ for a code $\cC_{\ins}(\ell,n)$ that detects up to $\ins$ insertions.
\end{enumerate*}
\end{definition}

\section{Detecting Deletions}\label{sec:detecting_del}
In this section, we present an explicit code $\cD_{\del}(\ell,n)$ that detects up to $\del$ deletions in each of the concatenated blocks of a codeword $\x \in \cD_{\del}(\ell,n)$. We also derive converse results for this problem, which give a lower bound on the redundancy of codes $\cC_{\del}(\ell,n)$ that detect to up to $\del$ deletions per block.

\subsection{Results}
\iflong
In Theorem~\ref{thm1}, we state our result on the explicit code $\cD_{\del}(\ell,n)$ which we construct in Section~\ref{del:cons}. In Theorem~\ref{thm2}, we give a lower bound on the redundancy of any code $\cC_{\del}(\ell,n)$ that detects to up to $\del$ deletions per block. In Theorem~\ref{thm3}, we specifically consider block-by-block decdodable codes (Definition~\ref{def3}), and give a lower bound on the redundancy of such codes. The proofs of these theorems are given in the subsequent sections.
\fi
\begin{theorem}
\label{thm1}
For $\del,\ell,n\in \mathbb{Z}^+$, with $2\del< \ell \leq n/2$, let $$\mathbf{x}=\langle \mathbf{x}^1,\mathbf{x}^2,\ldots,\mathbf{x}^{n/\ell}\rangle \in \cD_{\del}(\ell,n).$$ Suppose that $\mathbf{x}$ is affected by at most $\del$ deletions in each of its blocks $\mathbf{x}^1,\ldots,\mathbf{x}^{n/\ell}$. The code $\cD_{\del}(\ell,n)$ given in Construction~\ref{cons1} detects up to $\del$ deletions per block. This code is encodable and block-by-block decodable in linear time $\cO(n)$, and its redundancy is $(2\del+1)(n/\ell-1)$ bits.
\end{theorem}

\begin{theorem}
\label{thm2}
For $\del,\ell,n\in \mathbb{Z}^+$, with $2\del< \ell \leq n/3$,  the redundancy $r_\del(\ell,n)$ of any code $\cC_{\del}(\ell,n)$ that detects up to $\del$ deletions per block satisfies $$r_\del(\ell,n) \geq 2\del(n/\ell -1) + \varepsilon(n/\ell -2),$$
where $\varepsilon = 2\del - \log (2^{2\del} -1)>0$.
\end{theorem}

\begin{theorem}
\label{thm3}
For $\del,\ell,n\in \mathbb{Z}^+$, with $2\del< \ell \leq n/2$,  the redundancy $r_\del(\ell,n)$ of a block-by-block decodable code $\cC_{\del}(\ell,n)$ that detects up to $\del$ deletions per block satisfies $$r_\del(\ell,n) \geq (2\del+1)(n/\ell -1).$$
\end{theorem}

{\em Discussion:} Theorem~\ref{thm1} shows that the code $\cD_{\del}(\ell,n)$ which we present in Construction~\ref{cons1} is efficiently encodable and decodable and has redundancy $(2\del+1)(n/\ell-1)$ bits. As we will show in Section~\ref{del:cons}, the total redundancy of $\cD_{\del}(\ell,n)$ is an aggregate of: \begin{enumerate*}[label={\textit{(\roman*)}}] \item $2\del+1$ redundant bits per block $\x^j$, for $j\in[2,n/\ell-1]$; \item~$\del$ redundant bits for $\x^1$; \item $\del+1$ redundant bits for $\x^{n/\ell}$. \end{enumerate*} This means that the redundancy per block only depends on $\del$, and is constant in terms of the size of the block $\ell$ and the size of the codeword $n$. 
Note that in case we want to also {\em correct} $\del$ deletions per block, then the redundancy per block needs to be $\Omega(\del\log (\ell/\del))$~\cite{L66}, i.e., at least logarithmic in $\ell$. 

Furthermore, $\cD_{\del}(\ell,n)$ is block-by-block decodable, so it follows from Theorem~\ref{thm3} that $\cD_{\del}(\ell,n)$ has {\em optimal} redundancy among all block-by-block decodable codes. Theorem~\ref{thm2} gives a lower bound on the redundancy of any code $\cC_{\del}(\ell,n)$ (not necessarily block-by-block decodable) that detects up to $\del$ deletions per block for $n/\ell \geq 3$. By comparing this lower bound to Theorem~\ref{thm1}, it is easy see that our code $\cD_{\del}(\ell,n)$ has an asymptotically optimal redundancy in $\del$ among all codes $\cC_{\del}(\ell,n)$ with $n/\ell \geq 3$.

\subsection{Code Construction}
\label{del:cons}
The code that we present for detecting up to $\del$ deletions per block is given by the following construction.
\begin{construction}[Code detecting up to $\del$ deletions]
\label{cons1}
For all $\del,\ell,n\in \mathbb{Z}^+$, with \mbox{$2\del< \ell \leq n/2$}, we define the following%
\begin{align*}
\mathcal{A}_{\del}^0(\ell) &\triangleq \big\{\mathbf{x}\in \mathbb{F}_2^{\ell}~\big|~\mathbf{x}_{[1,\del+1]}=\mathbf{0}^{\del+1}\big\},\\
\mathcal{A}_{\del}^1(\ell) &\triangleq \big\{\mathbf{x}\in \mathbb{F}_2^{\ell} ~\big|~ \mathbf{x}_{[\ell-\del+1,\ell]}=\mathbf{1}^{\del}\big\}.
\end{align*}
The code $\cD_{\del}(\ell,n)\subseteq \mathbb{F}_2^n$ is defined as the set
\begin{equation*}
\left\{
  \langle \mathbf{x}^1,\ldots,\mathbf{x}^{n/\ell}\rangle \;\middle|\;
  \begin{aligned}
  & \mathbf{x}^1\in \mathcal{A}_{\del}^1(\ell),\\
  & \mathbf{x}^j \in \mathcal{A}_{\del}^1(\ell) \cap \mathcal{A}_{\del}^0(\ell), \forall j \in [2, \frac{n}{\ell}-1], \\
  & \mathbf{x}^{n/\ell} \in \mathcal{A}_{\del}^0(\ell).
  \end{aligned}
\right\}.
\end{equation*}
\end{construction}
\iflong
Before we prove Theorem~\ref{thm1}, we will provide the steps of the decoding algorithm for the code $\cD_{\del}(\ell,n)$ and give an example for $\del=1$. The encoding algorithm is omitted since encoding can be simply done by setting the information bits to the positions in $[n]$ that are not restricted by Construction~\ref{cons1}. The proof of Theorem~\ref{thm1} is given in Section~\ref{proof1}.
\fi

%

\subsection{Decoding} 
\label{del:dec}
Suppose that $\mathbf{x}\in \cD_{\del}(\ell,n)$ (with $2\del< \ell \leq n/2$) is affected by at most $\del$ deletions in each block $\mathbf{x}^j$, for all $j\in[n/\ell]$, resulting in $\mathbf{y}\in \mathbb{F}_2^*$. The input of the decoder is $\mathbf{y}$, and since we consider deletions only, the output is \mbox{$(\delta_1,\delta_2,\ldots,\delta_{n/\ell})\in \mathbb{Z}_{\delta+1}^{n/\ell}$}, where $\delta_j \leq \del$ denotes the number of deletions that have occurred in block $j$. 

The decoding is performed on a block-by-block basis, so before decoding block $j$, we know that the previous $j-1$ blocks have been decoded correctly. Therefore, after decoding the first $j-1$ blocks, the decoder knows the correct starting position of block $j$ in $\mathbf{y}$. Let the starting position of block $j\in [n/\ell]$ in $\mathbf{y}$ be $\alpha_j$, with $\alpha_1=1$; and $\forall j\in [n/\ell -1]$ let $$\mathbf{s}^j \triangleq \mathbf{y}_{[\alpha_j+\ell-\del,\alpha_j+\ell-1]}.$$
To decode block $j$, the decoder scans the bits in $\mathbf{s}^j$ from left to right searching for the first occurrence of a $0$ (if any). If $\mathbf{s}^j$ has no zeros, i.e., \mbox{$\mathbf{s}^j=\mathbf{1}^{\del}$}, then the decoder declares that no deletions ($\delta_j=0$) have occurred in block $j$, and sets the starting position of block \mbox{$j+1$} to \mbox{$\alpha_{j+1}=\alpha_j+\ell$}. Else, if the first occurrence of a $0$ in $\mathbf{s}^j$ is at position $\beta_j$, with \mbox{$1\leq \beta_j \leq \del$}, then the decoder declares that $\del_j=\del-\beta_j+1$ deletions have occurred in block $j$, and sets the starting position of block $j+1$ to \mbox{$\alpha_{j+1}=\alpha_j+\ell-\del_j$}. The decoder repeats this process for each block until the first $n/\ell-1$ blocks are decoded. Finally, the decoder checks the length of the last block in $\mathbf{y}$ based on its starting position $\alpha_{n/\ell}$, and outputs $\del_{n/\ell}$ accordingly. Note that this decoder satisfies Definition~\ref{def3} since the index of the last bit in $\mathbf{s}^j$, $\alpha_j+\ell-1$, is $\ell$ positions away from the starting position of the block $j$, $\alpha_j$.
\iflong
\begin{remark} 
A code equivalent to $ \cD_{\del}(\ell,n)$ can be obtained by flipping all the zeros to ones and vice-versa in the positions that are restricted by Construction~\ref{cons1}. The decoding algorithm described above can be also modified accordingly. However, this decoding algorithm cannot be applied for a code that combines codewords from the two aforementioned constructions. 
\end{remark}
\fi

\subsection{Example for $\del=1$}
Next we give an example of our code for $\del=1, \ell=5,$ and $n=20$. Consider a codeword $\mathbf{x}\in \cD_{1}(5,20)$ given by $$\mathbf{x}=\overbrace{10\underline{1}0\sepr}^{\mathbf{x}^1}\overbrace{\sepb\sepb 11\sepr}^{\mathbf{x}^2}\overbrace{\sepb\sepb01\underline{\sepr}}^{\mathbf{x}^3}\overbrace{\underline{\sepb}\sepb 100}^{\mathbf{x}^4}.$$ Suppose that the bits underlined in $\mathbf{x}$ are deleted, resulting in $$\mathbf{y}=10010011100010100.$$ To determine the number of deletions in the first block, the decoder first examines $\mathbf{s}^1=y_5=0$, which implies that $\beta_1=1$, and therefore declares that $\delta_1=1$ deletion has occurred in $\mathbf{x}^1$. The starting position of block $2$ is thus set to $\alpha_2=5$. Then, since $\mathbf{s}^2=y_9=1$, the decoder declares that no deletions ($\delta_2=0$) have occurred in $\mathbf{x}^2$, and sets the starting position of block $3$ to $\alpha_3=10$. Similarly, we have $\mathbf{s}^3=y_{14}=0$, implying that $\delta_3=1$ deletion has occurred in $\mathbf{x}^3$, $\beta_3=1$, and $\alpha_4=14$. Now since $\ell-1=4$ bits are left for the last block, the decoder declares that $\delta_4=1$ deletion has occurred in $\mathbf{x}^4$. 

\subsection{Proof of Theorem~\ref{thm1}}
\label{proof1}
The redundancy of the code $\cD_{\del}(\ell,n)$ follows from the number of bits that are fixed in Construction~\ref{cons1} which is \mbox{$(2\del+1)(n/\ell-1)$}. Encoding can be simply done by setting the information bits to the positions in $[n]$ that are not restricted by Construction~\ref{cons1}.
Hence, the complexities of the encoding and decoding algorithms are $\cO(n)$ since they involve a single pass over the bits with constant time operations. Next, we prove the correctness of the decoding algorithm.

Consider a codeword $\mathbf{x} \in \cD_{\del}(\ell,n)$ that is affected by at most $\delta$ deletions per block resulting in a string~$\mathbf{y}$.  Since the decoding is done on a block-by-block basis, it is enough to prove the correctness of the algorithm for the case where we have $n/\ell=2$ concatenated blocks. To prove that the code $\cD_{\del}(\ell,n)$ detects up to $\del$ deletions per block, we use induction on $\del\in \mathbb{Z}^+$, with $2\del< \ell$ and $n=2\ell$.

{\em Base case:} we first prove correctness for $\del=1$. The codeword is given by $$\mathbf{x}=\langle \mathbf{x}^1,\mathbf{x}^2 \rangle=(x_1,\ldots,x_{n/2},x_{n/2+1},\ldots,x_n) \in \cD_1(n/2,n).$$ To decode the first block, the decoder observes $\mathbf{s}^1=y_{n/2}$. Now consider two cases: \begin{enumerate*}[label={\textit{(\roman*)}}] \item no deletion has occurred in $\mathbf{x}^1$; and \item one bit was deleted in $\mathbf{x}^1$. \end{enumerate*} In the first case, based on Construction~\ref{cons1}, we always have \mbox{$\mathbf{s}^1=y_{n/2}=1$}. Therefore, it follows from the decoding algorithm described in Section~\ref{del:dec}, that the decoder can correctly declare that no deletions have occurred in the first block. Now consider the second case mentioned above. It follows from the code construction that the values of the first two bits of $\mathbf{x}^2$ are both~$0$. Hence, it is easy to see that for any single deletion in $\mathbf{x}^1$, and for any single deletion in $\mathbf{x}^2$, we always have \mbox{$\mathbf{s}^1=y_{n/2}=0$}. Thus, the decoder can always correctly detect a single deletion in $\mathbf{x}^1$. The number of deletions in the second block is consequently determined based on the starting position of the second block in $\mathbf{y}$, and the number of bits in $\mathbf{y}$ that are yet to be decoded. This concludes the proof for $\del=1$.

{\em Inductive step:} we assume that the code $\cD_{\del-1}(n/2,n)$ detects up to $\del-1$ deletions per block, and prove that $\cD_{\del}(n/2,n)$ detects up to $\del$ deletions per block. Based on Construction~\ref{cons1}, it is easy to see that if \mbox{$x\in \cD_{\del}(n/2,n)$}, then \mbox{$x \in \cD_{\del-1}(n/2,n)$}, and hence \mbox{$\cD_{\del}(n/2,n) \subset  \cD_{\del-1}(n/2,n)$}. Therefore, it follows from the inductive hypothesis that $\cD_{\del}(n/2,n)$ detects up to $\del-1$ deletions per block. Next, we prove that $\cD_{\del}(n/2,n)$ can also detect exactly $\del$ deletions per block. Suppose that exactly $\del$ deletions occur in~$\mathbf{x}^1$. To decode the first block, the decoder scans the bits of $\mathbf{s}^1=\mathbf{y}_{[\ell-\del+1,\ell]}$ from left to right searching for the first occurrence of a $0$, as explained in Section~\ref{del:dec}. It follows from the code construction that for any $\del$ deletions in $\mathbf{x}^1$, the first bit of $\mathbf{s}^1$ is always $0$, i.e., $y_{\ell-\del+1}=0$. To see this, notice that: \begin{enumerate*}[label={\textit{(\roman*)}}] \item for any $\del$ deletions in $\mathbf{x}^1$, the first bit belonging to the second block in $\mathbf{y}$ will shift $\del$ positions to the left; and \item since $\mathbf{x}_{[n/2+1,n/2+\del+1]}=\mathbf{0}^{\del+1}$, and given that we consider at most $\del$ deletions in $\mathbf{x}^2$, then the first bit belonging to the second block in $\mathbf{y}$ is always a $0$. \end{enumerate*} Hence, for any $\del$ deletions in $\mathbf{x}^1$, and for any $\del$ or fewer deletions in $\mathbf{x}^2$, we have $y_{\ell-\del+1}=0$. It follows from the decoding algorithm that the decoder in this case declares that $\del$ deletions have occurred in~$\mathbf{x}^1$. Furthermore, the number of deletions in the second block is consequently determined based on the starting position of the second block in $\mathbf{y}$, and the number of bits in $\mathbf{y}$ that are yet to be decoded. Therefore, we have proved that the code $\cD_{\del}(\ell,n)$ detects up to $\del$ deletions per block.

\subsection{Proofs of Theorem~\ref{thm2} and Theorem~\ref{thm3}}
The proofs of Theorem~\ref{thm2} and Theorem~\ref{thm3} have a common part where we show that for any codeword $\x \in \cC_\del(\ell,n)$, the last $\del$ bits of the blocks $\x^1,\dots, \x^{n/\ell -1}$ and the first $\del$ bits of blocks $\x^2,\dots,\x^{n/\ell}$ must be predetermined (fixed). This gives a preliminary lower bound on the redundancy of any code $\cC_\del(\ell,n)$, with $2\del< \ell \leq n/2$, that is 
\begin{equation}
\label{eq:bound1}
r_\del(\ell,n)\geq 2\del(n/\ell-1).
\end{equation}  
Then, to obtain the result in Theorem~\ref{thm2}, we improve the bound in~\eqref{eq:bound1} for $n/\ell \geq 3$ by showing that for any code $\cC_\del(\ell,n)$ an additional constraint must be imposed on each block $\x^j$ for $j\in [n/\ell-2]$, which gives 
\begin{equation}
\label{eq:bound2}
r_\del(\ell,n)\geq 2\del(n/\ell -1) + \varepsilon(n/\ell -2),
\end{equation} 
for some $0<\varepsilon<1$. For Theorem~\ref{thm3}, we specifically consider block-by-block decodable codes defined in Definition~\ref{def3}. We show that for such codes the bound in~\eqref{eq:bound1} can be improved to
\begin{equation}
\label{eq:bound3}
r_\del(\ell,n) \geq (2\del+1)(n/\ell -1).
\end{equation}
\iflong
\begin{figure*}
\begin{align}\label{eq:combi_1}
\y_1 &= (\cdots,{x_{\ell-\del}^j},{\color{red}x_{\del+1}^{j+1}},\ {\color{red}\cdots},\ {\color{red}x_{3\del}^{j+1}} \ ,x_{3\del+1}^{j+1},\cdots, x_{5\del}^{j+1}, \cdots \cdots \cdots, {\color{violet}x_{\ell-3\del+1}^{j+1}},{\color{violet}\cdots }, {\color{violet}x_{\ell-\del}^{j+1}},\ {\color{blue} x_{\ell-\del+1}^{j+1}},{\color{blue}\cdots}, {\color{blue} x_\del^{j+2}}, x_{\del+1}^{j+2}, \cdots ). \\ 
\y_2 &= (\cdots,{x_{\ell-\del}^j},{\color{blue}x_{\ell-\del+1}^{j}},{\color{blue}\cdots},{\color{blue}x_{\del}^{j+1}},{\color{red}x_{\del+1}^{j+1}},{\color{red}\cdots},{\color{red}x_{3\del}^{j+1}},\cdots \cdots \cdots, {\color{black}x_{\ell-5\del+1}^{j+1}},\cdots, {\color{black}x_{\ell-3\del}^{j+1}}, {\color{violet}x_{\ell-3\del+1}^{j+1}},{\color{violet}\cdots }, {\color{violet}x_{\ell-\del}^{j+1}},x_{\del+1}^{j+2}, \cdots). \label{eq:combi_2}
\end{align}
\end{figure*}
\fi

We start by showing that for every codeword $\x\in \cC_\del(\ell,n)$, each bit in the last $\delta$ bits $x_{\ell-\delta+1}^j,\ldots,x_\ell^j$ of any block $\x^j$, \mbox{$j\in [n/\ell -1]$}, must be different than all $\del$ bits $x_{1}^{j+1},\dots,x_{\del}^{j+1}$ of block $\x^{j+1}$. Suppose that for some $i_1 \in [\ell-\del+1,\ell]$ and $i_2 \in [\del]$, we have $x_{i_1}^j=x_{i_2}^{j+1}$. Consider the deletion combination where the last $\ell-i_1+1$ bits $x_{i_1}^j,\ldots,x_\ell^j$ are deleted in $\mathbf{x}^j$, and the first $i_2-1$ bits $x_1^{j+1},\ldots,x_{i_2-1}^{j+1}$ are deleted in $\x^{j+1}$ (if $i_2=1$ no bits are deleted in $\x^{j+1}$). The resulting string is of the form
\begin{equation*}
\y_1 = (\cdots,x_{i_1-1}^j,{\color{red}x_{i_2}^{j+1}},x_{i_2+1}^{j+1},\cdots, x_{\del}^{j+1},\cdots, x_\ell^{j+1}, \cdots),
\end{equation*}
where all other blocks are not affected by deletions. Now consider a different deletion combination where the last $\ell-i_1$ bits $x_{i_1+1}^j,\ldots,x_\ell^j$ are deleted in $\mathbf{x}^j$, and the first $i_2$ bits $x_1^{j+1},\ldots,x_{i_2}^{j+1}$ are deleted in $\x^{j+1}$ (if $i_1=\ell$ no bits are deleted in $\x^{j}$). The resulting string is of the form
\begin{equation*}
\y_2 = (\cdots,x_{i_1-1}^j,{\color{blue}x_{i_1}^{j}},x_{i_2+1}^{j+1},\cdots, x_{\del}^{j+1},\cdots, x_\ell^{j+1}, \cdots),
\end{equation*}
where all other blocks are not affected by deletions. Since $x_{i_1}^{j}=x_{i_2}^{j+1}$ by assumption, then we have $\y_1=\y_2$. %
\iflong
However, $\y_1$ and $\y_2$ correspond to two different deletion combinations (i.e., two different decoder outputs) given by
\begin{align*}
 \mathrm{Dec}(\y_1) &=(0,\dots,0,\underbrace{\ell-i_1 +1}_{j},\underbrace{i_2-1}_{j+1},0,\dots,0) \in \mathbb{Z}_{\delta+1}^{n/\ell},  \\
 \mathrm{Dec}(\y_2) &= (0,\dots,0,\underbrace{\ell-i_1}_{j},\underbrace{i_2}_{j+1},0,\dots,0) \in \mathbb{Z}_{\delta+1}^{n/\ell}.
\end{align*}
\fi

Therefore, if for any $i_1 \in [\ell-\del+1,\ell]$ and $i_2 \in [\del]$ we have $x_{i_1}^{j}=x_{i_2}^{j+1}$, then there exists two different decoder outputs that correspond to the same decoder input, i.e., $\y_1=\y_2$ with $\mathrm{Dec}(\y_1)\neq \mathrm{Dec}(\y_2)$. This contradicts the definition of the decoding function given in Definition~\ref{def2}. We conclude that the following conditions are necessary: 
\begin{equation}\label{eq:condition_2}
    x_{i_1}^j\neq x_{i_2}^{j+1},\ \forall \ i_1\in [\ell-\del+1,\ell], i_2\in[\del], \text{and } j\in [n/\ell -1].
\end{equation}

Since we focus on binary codes, the previous constraints imply that the last $\del$ bits of every block $\x^j$, \mbox{$j\in [n/\ell-1]$}, must be equal and must also be different than the first $\del$ bits of the block $\x^{j+1}$, which also should be equal. Namely, the code $\cC_{\del}(\ell,n)$ can have codewords that either satisfy \begin{enumerate*}[label={\textit{(\roman*)}}] \item $\x^j_{[\ell-\del+1,\ell]}=\mathbf{1}^{\del}$ and $\x^{j+1}_{[\del]}=\mathbf{0}^{\del}$; or \item $\x^j_{[\ell-\del+1,\ell]}=\mathbf{0}^{\del}$ and $\x^{j+1}_{[\del]}=\mathbf{1}^{\del}$. \end{enumerate*}
\iflong
Next, we show that although the previous statement is true, the code $\cC_{\del}(\ell,n)$ cannot have a pair of codewords $\x_1,\x_2\in \cC_{\del}(\ell,n)$, where $\x_1$ satisfies \textit{(i)} and $\x_2$ satisfies \textit{(ii)}. To prove this, we suppose that there exists such a pair $\x_1,\x_2\in \cC_{\del}(\ell,n)$, and then show that in this case we have two different decoder outputs that correspond to the same input string $\y$, which contradicts the definition of the decoding function in Definition~\ref{def2}.

Consider the following two deletion combinations for a given \mbox{$j\in [n/\ell-1]$}. In the first one, $\del-1$ out of the last $\del$ bits are deleted in $\x_1^j$, the first $\del$ bits are deleted in $\x_1^{j+1}$, and no bits are deleted in other blocks. In the second one, the last $\del$ bits are deleted in $\x_2^{j}$, $\del-1$ out of the first $\del$ bits are deleted in $\x_2^{j+1}$, and no bits are deleted in other blocks. In both combinations, the resulting string is
\begin{equation*}
\y_1 = \y_2 = (\cdots,{x_{\ell-\del}^j},{\color{red}1},x_{\del+1}^{j+1},x_{\del+2}^{j+1},\cdots, x_\ell^{j+1}. \cdots),
\end{equation*}
Since $\y_1$ and $\y_2$ correspond to two different decoder outputs
\begin{align*}
 \mathrm{Dec}(\y_1) &=(0,\dots,0,\underbrace{\del-1}_{j},\underbrace{\del}_{j+1},0,\dots,0) \in \mathbb{Z}_{\delta+1}^{n/\ell},  \\
 \mathrm{Dec}(\y_2) &= (0,\dots,0,\underbrace{\del}_{j},\underbrace{\del-1}_{j+1},0,\dots,0) \in \mathbb{Z}_{\delta+1}^{n/\ell},
\end{align*}
then we have $\y_1=\y_2$ with $\mathrm{Dec}(\y_1)\neq \mathrm{Dec}(\y_2)$ which contradicts Definition~\ref{def2}. 
\fi

Hence, we conclude that the first $\del$ bits and the last $\del$ bits in every block $\x^j$, $j\in [2,n/\ell-1]$, must be predetermined, and the last $\del$ bits of $\x^1$ and the first $\del$ bits of $\x^{n/\ell}$ must be predetermined. Therefore, $$|\cC_\del(\ell,n)| \leq 2^{n - 2\del(n/\ell -1)},$$ which gives the lower bound on the redundancy in~\eqref{eq:bound1}.



\iflong
Given the aforementioned constraints, we show next that for $n/\ell \geq 3$, the bound in~\eqref{eq:bound1} is not achievable, i.e., a redundancy of exactly $2\del(n/\ell -1)$ bits is not sufficient for decoding. 

Consider the following two deletion combinations for a given $j\in[n/\ell-2]$. In the first one, the last $\del$ bits are deleted in $\x^j$, the first $\del$ bits are deleted in $\x^{j+1}$, and no bits are deleted in other blocks. The resulting string $\y_1$ is of the form given in~\eqref{eq:combi_1}. In the second one, the last $\del$ bits are deleted in  $\x^{j+1}$, the first $\del$ bits are deleted in $\x^{j+2}$, and no bits are deleted in other blocks. The resulting string $\y_2$ is of the form given in~\eqref{eq:combi_2}. As explained previously, since the deletion combinations are different, it must hold that $\y_1\neq \y_2$. Since the bits $x_{\ell-\del+1}^{j},\cdots,x_{\del}^{j+1}$ are fixed, it follows from~\eqref{eq:combi_1} and~\eqref{eq:combi_2} that the condition $\y_1\neq \y_2$ is equivalent to  $({\color{red}x_{\del+1}^{j+1}},{\color{red}\cdots},{\color{red}x_{3\del}^{j+1}})\neq ({\color{blue}x_{\ell-\del+1}^{j}},{\color{blue}\cdots},{\color{blue}x_{\del}^{j+1}})$. This additional constraint that must be imposed on the $n/\ell-2$ blocks $\x^2,\ldots,\x^{n/\ell-1}$ introduces an additional redundancy in each of these blocks of value $\varepsilon = 2\del - \log (2^{2\del} -1)$, where $0<\varepsilon<1$. Consequently, we obtain the bound in~\eqref{eq:bound2} which concludes the proof of Theorem~\ref{thm2}. 
\fi

Next, we specifically consider block-by-block decodable codes and show that for such codes, in addition to the constraints in~\eqref{eq:condition_2}, the following must hold
\begin{equation}
\label{eq:condition_3}
x_{\del+1}^{j+1} \neq x_{i_1}^j, \forall \ i_1\in [\ell-\del+1,\ell] \text{ and } j\in [n/\ell -1].
\end{equation}
Namely, the constraint in~\eqref{eq:condition_3} extends the necessary conditions in~\eqref{eq:condition_2} to $i_2 \in [\del+1]$. We know from~\eqref{eq:condition_2} that the bits $x_{\ell-\del+1}^j,\ldots,x_{\ell}^j$ have the same values for all $j\in [n/\ell -1]$. Without loss of generality, assume that $\x^j_{[\ell-\delta+1,\ell]}=\mathbf{1}^{\del}$, and suppose that $x_{\del+1}^{j+1}=1$ has the same value as these bits. Consider the following two deletion combinations. In the first one, only one out of the last $\del$ bits is deleted in $\x^j$, and the first $\del$ bits are deleted in $\x^{j+1}$. In the second one, no bits are deleted in $\x^j$. Let $\alpha_j$ be the starting position of block $j$ in $\y$, and assume that the value of $\alpha_j$ is known at the decoder. For both combinations, we have $$\y_{[\alpha_j,\alpha_j+\ell-1]}=\langle \x^j_{[\ell-\del]},\mathbf{1}^{\del}\rangle.$$ Hence, a decoder cannot determine the exact number of deletions in $\x^j$ (0 or 1) by only processing the $\ell$ bits in $\y_{[\alpha_j,\alpha_j+\ell-1]}$. This contradicts the definition of block-by-block decodable codes given in Definition~\ref{def3}. Therefore, the condition in~\eqref{eq:condition_3} is necessary for all block-by-block decodable codes. This condition introduces an additional redundancy of $1$ bit in each of the $n/\ell-1$ blocks $\x^1,\ldots,\x^{n/\ell-1}$. Thus, to conclude the proof of Theorem~\ref{thm3}, we add $n/\ell -1$ to the RHS in~\eqref{eq:bound1} and obtain the bound in~\eqref{eq:bound3}.

\section{Detecting Insertions}\label{sec:detecting_ins}
In this section, we study the case where a string is affected by insertions only. We introduce two codes $\mathcal{I}_{1}(\ell,n)$ and $\mathcal{I}_{2}(\ell,n)$ that can detect up to $1$ and up to $2$ insertions per block, respectively. The problem of constructing codes that can detect up to $\ins> 2$ insertions is more involved. This will be already clear from the decoder of $\mathcal{I}_{2}(\ell,n)$ presented in Section~\ref{sec:detecting_ins}-C.

\subsection{Insertion Model}
\label{insmodel}
Consider $\mathbf{x}=\langle \mathbf{x}^1,\mathbf{x}^2,\ldots,\mathbf{x}^{n/\ell}\rangle \in \mathbb{F}_2^n$ that is affected by at most $\ins$ insertions in each block resulting in $\y \in \mathbb{F}_2^*$. In this section, we are interested in constructing codes that can detect up to $\ins$ insertions in each of the blocks  $\mathbf{x}^1,\ldots,\mathbf{x}^{n/\ell}$ of size $\ell$. We say that an insertion has occurred in block $j$, as opposed to block $j+1$, if the inserted bit appears before $x^j_{\ell}$ in $\mathbf{y}$, for all $j \in [n/\ell-1]$. Furthermore, unlike the case of deletions, some insertion combinations cannot be distinguished by any code that detects insertions. This situation arises in particular cases where different insertion combinations generate the same string at the block boundaries.
\iflong
For instance, consider the following two insertion combinations for a given $j\in [n/\ell-2]$. In the first one, $x^j_{\ell}$ is inserted at position $\ell$ in $\x^j$, $x^{j+1}_{\ell}$ is inserted at position $\ell$ in $\x^{j+1}$, and no bits are inserted in the other blocks. In the second one, $x^j_{\ell}$ is inserted at position $1$ in $\x^{j+1}$, $x^{j+1}_{\ell}$ is inserted at position $1$ in $\x^{j+2}$, and no bits are inserted in the other blocks. For both combinations, we obtain
\begin{equation}
\label{eq:conv}
\y = (\cdots,x_1^j,\cdots,x_{\ell}^j,x_{\ell}^j,x_1^{j+1},\cdots \cdots, x_{\ell}^{j+1},x_{\ell}^{j+1},x_1^{j+2}\cdots).
\end{equation}
Therefore, in such particular cases, no code can detect whether the first insertion actually occurred in block $j$ or block $j+1$. To this end, we adopt the following decoding convention. Consider an example where a string $100$ is affected by $2$ insertions resulting in $10100$. Notice that there are multiple insertion combinations that could have generated $10100$ from $100$. For instance, the bits $10$ could have been inserted at the beginning (i.e., $\underline{1}\underline{0}100$), or $01$ could have been inserted in the second position (i.e., $1\underline{0}\underline{1}00$), etc.
\fi
\iflong
For such cases, our decoding convention is to assume that the actual insertion combination is the one that occurs in the leftmost position in the string. In the previous example, we assume that $10$ was inserted at the beginning. If we apply this convention to the case discussed in~\eqref{eq:conv}, then the decoder would declare that one insertion occurred in each of $\x^j$ and $\x^{j+1}$, and no insertions occurred in other blocks.
\fi

\subsection{One Insertion}
Theorem~\ref{thmins1} shows that the code $\mathcal{I}_{1}(\ell,n)$ (defined in Construction~\ref{cons2}) detects up to $1$ insertion in each block, with redundancy $2(n/\ell-1)$.
\begin{construction}[Code detecting up to $1$ insertion]
\label{cons2}
For $\ell,n\in \mathbb{Z}^+$, with $2 < \ell \leq n/2$, we define 
\begin{align*}
\mathcal{B}_1^0(\ell) &\triangleq \big\{\mathbf{x}\in \mathbb{F}_2^{\ell}~\big|~x_1=0\big\}, \\
\mathcal{B}_1^1(\ell) &\triangleq \big\{\mathbf{x}\in \mathbb{F}_2^{\ell}~\big|~x_{\ell}=1\big\}.
\end{align*}
The code $\mathcal{I}_1(\ell,n)$ is defined as the set
\begin{equation*}
\left\{
  \langle \mathbf{x}^1,\ldots,\mathbf{x}^{n/\ell}\rangle \;\middle|\;
  \begin{aligned}
  & \mathbf{x}^1\in \mathcal{B}_{1}^1(\ell),\\
  & \mathbf{x}^j \in \mathcal{B}_{1}^1(\ell) \cap \mathcal{B}_{1}^0(\ell), \forall j \in [2, \frac{n}{\ell}-1], \\
  & \mathbf{x}^{n/\ell} \in \mathcal{B}_{0}^0(\ell).
  \end{aligned}
\right\}.
\end{equation*}
\end{construction}

\begin{theorem}
\label{thmins1}
For $\ell,n\in \mathbb{Z}^+$, with $2 < \ell \leq n/2$, consider a codeword $\mathbf{x}=\langle \mathbf{x}^1,\mathbf{x}^2,\ldots,\mathbf{x}^{n/\ell}\rangle \in \mathcal{I}_1(\ell,n)$ that is affected by at most $1$ insertion in each of its blocks $\mathbf{x}^1,\ldots,\mathbf{x}^{n/\ell}$. The code $\mathcal{I}_1(\ell,n)$ detects up to $1$ insertion per block.
This code is encodable and block-by-block decodable in linear time $\cO(n)$, and its redundancy is $2(n/\ell-1)$.
\end{theorem}

\iflong
\begin{proof}
The complexity and redundancy arguments are similar to the ones in the proof of Theorem~\ref{thm1}. Next, we describe the decoding algorithm and prove its correctness. Consider a codeword $\mathbf{x} \in \cI_{1}(\ell,n)$ that is affected by at most $1$ insertion  per block, resulting in~$\mathbf{y}\in \mathbb{F}_2^*$. The input of the decoder is~$\mathbf{y}$, and the output is the number of insertions in each block~(either $0$ or $1$). 

The decoding is done a block-by-block basis. Assume that the first $j-1$ blocks were decoded correctly, so the decoder knows the starting position of block $j$ in $\mathbf{y}$, denoted by $\alpha_j$. To decode block $j$, the decoder examines bit $y_{\alpha_j+\ell+1}$. If $y_{\alpha_j+\ell+1}=0$, then the decoder declares that no insertions have occurred, and sets the starting position of block $j+1$ to $\alpha_{j+1}=\alpha_j+\ell$. Else, if $y_{\alpha_j+\ell+1}=1$, then the decoder declares that one insertion has occurred, and sets the starting position of block $j+1$ to $\alpha_{j+1}=\alpha_j+\ell+1$. The decoder repeats this process for the first $n/\ell-1$ blocks, and the number of insertions in the last block is deduced from $\alpha_{n/\ell}$ and the number of bits in $\mathbf{y}$ that are yet to be decoded.

It follows from Construction~\ref{cons2}, and the insertion model described in Section~\ref{insmodel}, that: 
\begin{enumerate*}[label={\textit{(\roman*)}}]
\item $y_{\alpha_j+\ell+1}=0$ only if no bits were inserted  $\x^j$; and \item $y_{\alpha_j+\ell+1}=1$ only if a single bit was inserted in $\x^j$ or the bit $1$ was inserted at the beginning of $\x^{j+1}$. \end{enumerate*}
However, based on our decoding convention described previously, a single insertion is declared in block $j$ for both cases. 
\end{proof}
\fi

\subsection{Two Insertions}
Interestingly, the problem becomes significantly more complex for $\ins\geq 2$ insertions. Next, we present a code construction for $\ins=2$. Theorem~\ref{thmins2} shows that the code $\mathcal{I}_2(\ell,n)$ (defined in Construction~\ref{cons3}) detects up to $2$ insertions per block with redundancy $8n/\ell -5$.

\begin{construction}[Code detecting up to $2$ insertions]
\label{cons3}
For $\ell,n\in \mathbb{Z}^+$, with $8 < \ell \leq n/2$, we define
\begin{align*}
\mathcal{B}_2^0(\ell) &\triangleq  \big\{\mathbf{x}\in \mathbb{F}_2^{\ell}~\big|~\mathbf{x}_{[1,5]}=\langle\mathbf{0}^2,\mathbf{1}^3\rangle\big\}, \\
\mathcal{B}_2^1(\ell) &\triangleq  \big\{\mathbf{x}\in \mathbb{F}_2^{\ell}~\big|~\mathbf{x}_{[\ell-2,\ell]}=\langle 0,\mathbf{1}^2 \rangle\big\}.
\end{align*}
The code $\mathcal{I}_2(\ell,n)$ is defined as the set
\begin{equation*}
\left\{
  \langle \mathbf{x}^1,\ldots,\mathbf{x}^{n/\ell}\rangle \;\middle|\;
  \begin{aligned}
  & \mathbf{x}^1\in \mathcal{B}_{2}^1(\ell),\\
  & \mathbf{x}^j \in \mathcal{B}_{2}^1(\ell) \cap \mathcal{B}_{2}^0(\ell), \forall j \in [2, \frac{n}{\ell}].
  \end{aligned}
\right\}.
\end{equation*}
\end{construction}

\begin{theorem}
\label{thmins2}
For $\ell,n\in \mathbb{Z}^+$, with $8 < \ell \leq n/2$, consider a codeword $\mathbf{x}=\langle \mathbf{x}^1,\mathbf{x}^2,\ldots,\mathbf{x}^{n/\ell}\rangle \in \mathcal{I}_2(\ell,n)$ that is affected by at most $2$ insertions in each of its blocks $\mathbf{x}^1,\ldots,\mathbf{x}^{n/\ell}$. The code $\mathcal{I}_2(\ell,n)$ detects up to $2$ insertions per block. The code is encodable and decodable in linear time~$\cO(n)$, and its redundancy is $8n/\ell-5$ bits.
\end{theorem}

{\em Decoding:} Consider a codeword $\mathbf{x}\in \cI_{2}(\ell,n)$ (with \mbox{$8< \ell \leq n/2$}) that is affected by at most $2$ insertions in each block, resulting in $\mathbf{y}\in \mathbb{F}_2^*$. The input of the decoder is $\mathbf{y}$, and the output is \mbox{$(\ins_1,\ins_2,\ldots,\ins_{n/\ell})\in \mathbb{Z}_{3}^{n/\ell}$}, where $\ins_j$ denotes the number of insertions that have occurred in block $j$. The decoding is done on a block-by-block basis. The blocks are decoded sequentially from left to right, with a small amount of lookahead that is at most $\ell$ bits of the next block. The process through which the decoder determines its output $\ins_j$ for a block $j\in [n/\ell-1]$ is illustrated in Figure~\ref{fig:2ins}. As for the last block, the decoder determines $\ins_{n/\ell}$ based on the starting position of the last block $\alpha_{n/\ell}$ and the number of bits in $\mathbf{y}$ that are yet to be decoded.
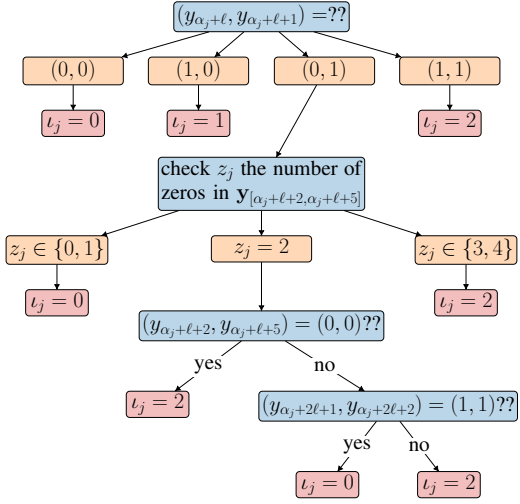
\begin{figure}
    \centering
    \resizebox{0.4\textwidth}{!}{
    \definecolor{color0}{rgb}{0.12156862745098,0.466666666666667,0.705882352941177}
\definecolor{color1}{rgb}{1,0.498039215686275,0.0549019607843137}
\definecolor{color2}{rgb}{0.172549019607843,0.627450980392157,0.172549019607843}
\definecolor{color3}{rgb}{0.83921568627451,0.152941176470588,0.156862745098039}
\definecolor{color4}{rgb}{0.580392156862745,0.403921568627451,0.741176470588235}
\definecolor{color6}{rgb}{0.549019607843137,0.337254901960784,0.294117647058824}
\definecolor{color5}{rgb}{0.890196078431372,0.466666666666667,0.76078431372549}

\tikzset{%
  >={Latex[width=2mm,length=2mm]},
            base/.style = {rectangle, rounded corners, draw=black,
                           minimum width=4cm, minimum height=1cm,
                           text centered, font=\Huge},
  activityStarts/.style = {base, fill=color0!30},
    activityRuns/.style = {base, fill=color1!30},
         process/.style = {base, minimum width=2.5cm, fill=color3!30,
                           font=\Huge},
         lemma/.style = { minimum width=2.5cm, ,
                           font=\huge},
}

\begin{tikzpicture}[node distance=5.5cm, 
    every node/.style={fill=white, font=\Large }, align=center]
  \node (start)             [activityStarts]              {$(y_{\alpha_j+\ell},y_{\alpha_j+\ell+1})=$??};
  \node (zz)     [activityRuns, below left = 1 cm and 2 cm of start]          {$(0,0)$};
  \node (oz)     [activityRuns, below left = 1 cm and -3 cm of start]          {$(1,0)$};
  \node (zo)     [activityRuns, below right = 1 cm and -3 cm of start]          {$(0,1)$};
    \node(oo)     [activityRuns, below right = 1 cm and 2 cm of start]          {$(1,1)$};

  \node (res1)     [process, below = 1 cm of zz]   {$\ins_j = 0$};
  \node (res2)     [process, below = 1 cm of oz]   {$\ins_j = 1$};
  \node (res3)     [process, below = 1 cm of oo]   {$\ins_j = 2$};
  \node (zj)     [activityStarts, below = 5 cm of start]          {\begin{varwidth}{8cm}check $z_j$ the number of zeros in $\y_{[\alpha_j+\ell+2,\alpha_j+\ell+5]}$\end{varwidth}};
  
  \node (01) [activityRuns, below left = 1 cm and 2 cm of zj]   {$z_j \in \{0,1\}$};
  \node (2) [activityRuns,  below = 1 cm  of zj]   {$z_j = 2$};
  \node (34) [activityRuns,below right = 1 cm and 2 cm of zj]   {$z_j \in \{3,4\}$};
  
  \node (res4)  [process, below = 1 cm of 01]   {$\ins_j = 0$};
  \node (res5)  [process, below = 1 cm of 34]   {$\ins_j = 2$};
  
    \node (zj2)     [activityStarts, below = 4 cm of zj]          {\begin{varwidth}{10cm}$(y_{\alpha_j+\ell+2},y_{\alpha_j+\ell+5})= (0,0)$??\end{varwidth}};
    
  \node (res6) [process, below left = 2 cm and -2 cm of zj2]   {$\ins_j = 2$};
  
  \node (no1) [activityStarts, below right = 2 cm and -5 cm of zj2]   {\begin{varwidth}{12cm}$(y_{\alpha_j+2\ell+1},y_{\alpha_j+2\ell+2})= (1,1)$??\end{varwidth}};
  
  \node (res7) [process, below left = 2 cm and -4 cm of no1]   {$\ins_j = 0$};
  \node (res8) [process, below right = 2 cm and -4 cm of no1]   {$\ins_j = 2$};
  
  \path[->]
  (start) edge (zz)
          edge (oz)
          edge (zo)
          edge (oo)
  (zz)    edge (res1)
  (oz)    edge (res2)
  (oo)    edge (res3)
  (zo)    edge (zj)
  (zj)    edge (01)
          edge (2)
          edge (34)
  (01)    edge (res4)
  (2)     edge (zj2)
  (34)    edge (res5)
  (zj2)   edge node[font=\Huge] {yes} (res6)
  (zj2)   edge node[font=\Huge] {no} (no1)
  (no1)   edge node[font=\Huge] {yes} (res7)
  (no1)   edge node[font=\Huge] {no} (res8)
        ;

\end{tikzpicture}
    }
    \caption{Decoding algorithm of $\mathcal{I}_2(\ell,n)$ that outputs the number of insertions in block $j\in [n/\ell-1]$ denoted by $\ins_j$, based on the input string $\y$. After decoding block $j$, the decoder sets the starting position of block $j+1$ in $\y$ to $\alpha_{j+1}=\alpha_j+\ell+\ins_j$, where $\alpha_1=1$, and proceeds as depicted in the figure.}
    \label{fig:2ins}
\end{figure}
\iflong
\begin{proof}[Proof of Theorem~\ref{thmins2}]
The encoding complexity and redundancy arguments are similar to the ones in the proof of Theorem~\ref{thm1}. Decoding a given block involves a single pass over $\cO(\ell)$ bits with constant time operations. Therefore, the decoding complexity is $\cO(n)$. Next, we prove the correctness of the decoding algorithm.

Given that at most $2$ insertions occur in each block, it follows from Construction~\ref{cons3}, and the insertion model described in Section~\ref{insmodel}, that: 
\begin{enumerate}[leftmargin=*]
\item $(y_{\alpha_j+\ell},y_{\alpha_j+\ell+1})=(0,0)$ only if no bits were inserted in~$\x^j$.
\item $(y_{\alpha_j+\ell},y_{\alpha_j+\ell+1})=(1,0)$ only if a single bit was inserted in $\x^j$ or the bit $1$ was inserted at the beginning of block $\x^{j+1}$. Based on our decoding convention, a single insertion is declared in $\x^j$ for both cases.
\item $(y_{\alpha_j+\ell},y_{\alpha_j+\ell+1})=(1,1)$ only if $2$ bits were inserted in $\x^{j}$ or the bits $11$ were inserted at the beginning of $\x^{j+1}$. Based on our decoding convention, an insertion is declared in $\x^j$ for both cases.
\item $(y_{\alpha_j+\ell},y_{\alpha_j+\ell+1})=(0,1)$ could either indicate that \begin{enumerate*}[label={\textit{(\roman*)}}] \item $2$~bits were inserted in $\x^j$, where one of these bits is a $0$ inserted at position $\ell$ in $\x^j$; or \item no bits were inserted in $\x^j$ and the bit $1$ was inserted in position $2$ in $\x^{j+1}$. \end{enumerate*} Assume that case \textit{(i)} is the correct one and thus $\alpha_{j+1}=\alpha_j+\ell+2$ is the starting position of block $j+1$ in~$\y$. Given that we consider at most $2$ insertions in any block and that $\x^{j+1}_{[1,5]}=\langle\mathbf{0}^2,\mathbf{1}^3\rangle$ by construction, then we expect to see at least $2$ zeros in $(y_{\alpha_j+\ell+2},\ldots,y_{\alpha_j+\ell+5})$ if our assumption is correct. On the other hand, if our assumption is wrong, then a $0$ was wrongfully removed from the start of block $\x^{j+1}$ and we expect to see at most $2$ zeros in the aforementioned string, where the second zero could be a result of an additional insertion. Let $z_j$ be the number of zeros in $(y_{\alpha_j+\ell+2},\ldots,y_{\alpha_j+\ell+5})$. Based on the discussion above, we now know that: \begin{enumerate*}[label={\textit{(\alph*)}}] 
\item $z_j\in \{0,1\}$ only if case \textit{(i)} is true; and \item $z_j\in \{3,4\}$ only if case \textit{(ii)} is true. 
\end{enumerate*} If $z_j=2$ and $(y_{\alpha_j+\ell+2},y_{\alpha_j+\ell+3})=(0,0)$, then case \textit{(i)} is assumed by convention. In all other cases where $z_j=2$, the assumption that case \textit{(i)} is correct could be wrong since we could be dealing with case \textit{(ii)} with an additional insertion of a $0$ at $\alpha_j+\ell+2$ or $\alpha_j+\ell+3$. One can verify from the construction that case \textit{(ii)} is true with $z_j=2$ only if $(y_{\alpha_j+2\ell+1},y_{\alpha_j+2\ell+2})=(1,1)$. Furthermore, case \textit{(i)} is true with $z_j=2$ only if  $(y_{\alpha_j+2\ell+1},y_{\alpha_j+2\ell+2})\neq (1,1)$.
All cases not mentioned in the discussion above correspond to instances where the decoding convention is applied.
\end{enumerate}
\end{proof}
\fi

\iflong
\section{Detecting Deletions or Insertions}\label{sec:mix}
\iflong
In this section, we introduce a code $\cC_1(\ell,n)$ that can detect up to $\tau=1$ one error in each block, where the error could be either a deletion or an insertion.

\subsection{Error Model}
\label{insdelmodel}
Consider $\mathbf{x}=\langle \mathbf{x}^1,\mathbf{x}^2,\ldots,\mathbf{x}^{n/\ell}\rangle \in \mathbb{F}_2^n$ that is affected by at most $\tau=1$ error in each block, resulting in $\y \in \mathbb{F}_2^*$.  For all $j \in [n/\ell-1]$, we say that an error has occurred in block $j$, as opposed to block $j+1$, if: \begin{enumerate*}[label={\textit{(\roman*)}}]  
\item a bit is deleted in $\mathbf{x}^j$; or \item a bit is inserted and appears before $x^j_{\ell}$ in $\mathbf{y}$.
\end{enumerate*}
For decoding, we adopt the convention that was explained for the case of insertions in Section~\ref{insmodel}. Furthermore, there are two particular deletion/insertion combinations that cannot be distinguished by any code that follows Definition~\ref{def2}, and whose purpose is to detect a single error per block. These combinations result in a transposition error at the block boundaries, where the positions of the last bit in $\mathbf{x}^j$ and the first bit in $\mathbf{x}^{j+1}$ are swapped. In fact, a transposition error which transforms $ab \to ba$ could be either a result of \begin{enumerate*}[label={\textit{(\roman*)}}] \item deleting $a$ from $ab$ and then inserting $a$ at the end; or \item inserting $b$ at the beginning of $ab$ and then deleting the $b$ at the end. \end{enumerate*}
If such a transposition error occurs at the block boundaries, a decoder cannot detect whether the deletion occurred in block $j$ and the insertion in block $j+1$, or the other way around. Namely, for both combinations we obtain
For both combinations, we obtain
\begin{equation}
\label{eq:conv1}
\y = (\cdots \cdots,x_1^j, \cdots, x_1^{j+1},x_{\ell}^j, \cdots, x_{\ell}^{j+1},\cdots \cdots).
\end{equation}
For this reason, we adopt a decoding convention which assumes that a transposition error $ab\to ba$ occurs by inserting $b$ at the beginning of and deleting the $b$ at the end. If we apply this convention to the case discussed in~\eqref{eq:conv1}, then the decoder would declare that one insertion occurred in $\x^j$ and one deletion occurred in $\x^{j+1}$, i.e., assuming no errors have occurred in the other blocks, the output of the decoder based on Definition~\ref{def2} would be $$ \boldsymbol{\del}= (0,\ldots,\underbrace{0}_{j},\underbrace{1}_{j+1},0,\ldots,0),$$
$$\boldsymbol{\ins}= (0,\ldots,\underbrace{1}_{j},\underbrace{0}_{j+1},0,\ldots,0).$$

\subsection{Code Construction}
Theorem~\ref{thmmix} shows that the code $\mathcal{C}_{1}(\ell,n)$ (defined in Construction~\ref{consmix}) detects up to $1$ error in each block, with redundancy $6(n/\ell-1)$. 
\begin{construction}[Code detecting up to $1$ error]
\label{consmix}
For $\ell,n\in \mathbb{Z}^+$, with $6 < \ell \leq n/2$, we define
\begin{align*}
\mathcal{G}^0(\ell) &\triangleq  \big\{\mathbf{x}\in \mathbb{F}_2^{\ell}~\big|~\mathbf{x}_{[1,3]}= \mathbf{0}^3\big\}, \\
\mathcal{G}^1(\ell) &\triangleq  \big\{\mathbf{x}\in \mathbb{F}_2^{\ell}~\big|~\mathbf{x}_{[\ell-2,\ell]}=\langle 0,\mathbf{1}^2 \rangle\big\}.
\end{align*}
The code $\mathcal{C}_1(\ell,n)$ is defined as the set
\begin{equation*}
\left\{
  \langle \mathbf{x}^1,\ldots,\mathbf{x}^{n/\ell}\rangle \;\middle|\;
  \begin{aligned}
  & \mathbf{x}^1\in \mathcal{G}^1(\ell),\\
  & \mathbf{x}^j \in \mathcal{G}^1(\ell) \cap \mathcal{G}^0(\ell), \forall j \in [2, \frac{n}{\ell}-1], \\
  & \mathbf{x}^{n/\ell} \in \mathcal{G}^0(\ell).
  \end{aligned}
\right\}.
\end{equation*}
\end{construction}

\begin{theorem}
\label{thmmix}
For $\ell,n\in \mathbb{Z}^+$, with $6 < \ell \leq n/2$, consider a codeword $\mathbf{x}=\langle \mathbf{x}^1,\mathbf{x}^2,\ldots,\mathbf{x}^{n/\ell}\rangle \in \mathcal{C}_1(\ell,n)$ that is affected by at most $1$ error in each of its blocks $\mathbf{x}^1,\ldots,\mathbf{x}^{n/\ell}$, where the error could be either a deletion or an insertion. The code $\mathcal{C}_1(\ell,n)$ detects up to $1$ error per block. The code is encodable and block-by-block decodable in linear time~$\cO(n)$, and its redundancy is $6(n/\ell-1)$ bits.
\end{theorem}

{\em Decoding:} Consider a codeword $\mathbf{x}\in \cC_{1}(\ell,n)$ (with \mbox{$6< \ell \leq n/2$}) that is affected by at most $2$ insertions in each block, resulting in $\mathbf{y}\in \mathbb{F}_2^*$. The input of the decoder is $\mathbf{y}$, and the output is $$(\del_1,\del_2,\ldots,\del_{n/\ell}) \times (\ins_1,\ins_2,\ldots,\ins_{n/\ell}) \in \mathbb{Z}_{2}^{n/\ell} \times \mathbb{Z}_{2}^{n/\ell},$$  where $\del_j$ and $\ins_j$ denote the number of deletions and insertions that have occurred in block $j$, respectively. The decoding is done on a block-by-block basis. The process through which the decoder determines its output $(\del_j,\ins_j)$ for a block $j\in [n/\ell-1]$ is illustrated in Figure~\ref{fig:3mix}. As for the last block, the decoder determines $(\del_{n/\ell},\ins_{n/\ell})$ based on the starting position of the last block $\alpha_{n/\ell}$ and the number of bits in $\mathbf{y}$ that are yet to be decoded.

\begin{figure}
    \centering
    \resizebox{0.4\textwidth}{!}{
    \definecolor{color0}{rgb}{0.12156862745098,0.466666666666667,0.705882352941177}
\definecolor{color1}{rgb}{1,0.498039215686275,0.0549019607843137}
\definecolor{color2}{rgb}{0.172549019607843,0.627450980392157,0.172549019607843}
\definecolor{color3}{rgb}{0.83921568627451,0.152941176470588,0.156862745098039}
\definecolor{color4}{rgb}{0.580392156862745,0.403921568627451,0.741176470588235}
\definecolor{color6}{rgb}{0.549019607843137,0.337254901960784,0.294117647058824}
\definecolor{color5}{rgb}{0.890196078431372,0.466666666666667,0.76078431372549}

\tikzset{%
  >={Latex[width=2mm,length=2mm]},
            base/.style = {rectangle, rounded corners, draw=black,
                           minimum width=4cm, minimum height=1cm,
                           text centered, font=\Huge},
  activityStarts/.style = {base, fill=color0!30},
    activityRuns/.style = {base, fill=color1!30},
         process/.style = {base, minimum width=2.5cm, fill=color3!30,
                           font=\Huge},
         lemma/.style = { minimum width=2.5cm, ,
                           font=\huge},
}

\begin{tikzpicture}[node distance=5.5cm, 
    every node/.style={fill=white, font=\Large }, align=center]
  \node (start)             [activityStarts]              {$y_{\alpha_j+\ell+1}=1$??};
  \node(a1) [process, below left = 2 cm and 0cm of start] {$(\del_j,\ins_j)=(0,1)$};
  
  \node(a2) [activityStarts, below right = 2 cm and 0cm of start] {$y_{\alpha_j+\ell}=1$??};
  
  \node(a3) [process, below right = 2 cm and -1cm of a2] {$(\del_j,\ins_j)=(1,0)$};
  
   \node(a4) [activityStarts, below left = 2 cm and -1cm of a2] {$y_{\alpha_j+\ell-2}=0$??};
   
    \node(a5) [process, below left = 2 cm and -2 cm of a4] {$(\del_j,\ins_j)=(0,0)$};
     
   \node(a6) [process, below right = 2 cm and -2 cm of a4] {$(\del_j,\ins_j)=(1,0)$};
   
   \path[->]
   (start)   edge node[font=\Huge] {yes} (a1)
   (start)   edge node[font=\Huge] {no} (a2)
   (a2)   edge node[font=\Huge] {yes} (a4)
   (a2)   edge node[font=\Huge] {no} (a3)
   (a4)   edge node[font=\Huge] {yes} (a5)
   (a4)   edge node[font=\Huge] {no} (a6)
   ;


\end{tikzpicture}
    }
    \caption{Decoding algorithm of $\mathcal{C}_1(\ell,n)$ that outputs the number of deletions and insertions in block $j\in [n/\ell-1]$ denoted by~$(\del_j,\ins_j)$, based on the input string $\y$. After decoding block $j$, the decoder sets the starting position of block $j+1$ in $\y$ to $\alpha_{j+1}=\alpha_j+\ell+\ins_j-\del_j$, where $\alpha_1=1$, and proceeds as depicted in the figure.}
    \label{fig:3mix}
\end{figure}

\else
\fi
\fi

\section{Conclusion}\label{sec:conc}
\iflong
In this paper, we studied the problem of constructing codes that detect the exact number of worst-case deletions and insertions in concatenated strings. First, we constructed codes that detect up to $\del$ deletions in each concatenated block. We derived fundamental limits for this problem which show that our codes are optimal among all block-by-block decodable codes, and asymptotically optimal in $\delta$ among all codes that detect up to $\del$ deletions. Then, we constructed two codes that detect up to $1$ and up to $2$ insertions per block. We also present a construction of a code that detects up to $1$ error per block, where the error could be either a deletion or an insertion. Some of the open problems include finding code constructions and fundamental limits for detecting up to $\ins>2$ deletions and up to $\tau>1$ errors. An additional interesting direction for future research is to consider non-binary codes for detecting deletions and insertions.
\fi

\bibliographystyle{IEEEtran}
\bibliography{Refs}

\begin{thebibliography}{10}
\providecommand{\url}[1]{#1}
\csname url@samestyle\endcsname
\providecommand{\newblock}{\relax}
\providecommand{\bibinfo}[2]{#2}
\providecommand{\BIBentrySTDinterwordspacing}{\spaceskip=0pt\relax}
\providecommand{\BIBentryALTinterwordstretchfactor}{4}
\providecommand{\BIBentryALTinterwordspacing}{\spaceskip=\fontdimen2\font plus
\BIBentryALTinterwordstretchfactor\fontdimen3\font minus
  \fontdimen4\font\relax}
\providecommand{\BIBforeignlanguage}[2]{{%
\expandafter\ifx\csname l@#1\endcsname\relax
\typeout{** WARNING: IEEEtran.bst: No hyphenation pattern has been}%
\typeout{** loaded for the language `#1'. Using the pattern for}%
\typeout{** the default language instead.}%
\else
\language=\csname l@#1\endcsname
\fi
#2}}
\providecommand{\BIBdecl}{\relax}
\BIBdecl

\bibitem{A18}
M.~{Abroshan}, R.~{Venkataramanan}, and A.~{Guillén i Fàbregas}, ``Coding for
  segmented edit channels,'' \emph{IEEE Transactions on Information Theory},
  vol.~64, no.~4, pp. 3086--3098, 2018.

\bibitem{L10}
Z.~{Liu} and M.~{Mitzenmacher}, ``Codes for deletion and insertion channels
  with segmented errors,'' \emph{IEEE Transactions on Information Theory},
  vol.~56, no.~1, pp. 224--232, 2010.

\bibitem{J20}
M.~{Cheraghchi}, R.~{Gabrys}, O.~{Milenkovic}, and J.~{Ribeiro}, ``Coded trace
  reconstruction,'' \emph{IEEE Transactions on Information Theory}, vol.~66,
  no.~10, pp. 6084--6103, 2020.

\bibitem{R15}
R.~{Venkataramanan}, V.~{Narasimha Swamy}, and K.~{Ramchandran},
  ``Low-complexity interactive algorithms for synchronization from deletions,
  insertions, and substitutions,'' \emph{IEEE Transactions on Information
  Theory}, vol.~61, no.~10, pp. 5670--5689, 2015.

\bibitem{L66}
V.~I. Levenshtein, ``Binary codes capable of correcting deletions, insertions
  and reversals,'' in \emph{Soviet physics doklady}, vol.~10, 1966, p. 707.

\bibitem{VT65}
R.~Varshamov and G.~Tenengol’ts, ``Correction code for single asymmetric
  errors,'' \emph{Automat. Telemekh}, vol.~26, no.~2, pp. 286--290, 1965.

\bibitem{B16}
J.~{Brakensiek}, V.~{Guruswami}, and S.~{Zbarsky}, ``Efficient low-redundancy
  codes for correcting multiple deletions,'' \emph{IEEE Transactions on
  Information Theory}, vol.~64, no.~5, pp. 3403--3410, May 2018.

\bibitem{GC}
S.~{Kas Hanna} and S.~{El Rouayheb}, ``Guess \& check codes for deletions,
  insertions, and synchronization,'' \emph{IEEE Transactions on Information
  Theory}, vol.~65, no.~1, pp. 3--15, Jan 2019.

\bibitem{H19}
B.~{Haeupler}, ``Optimal document exchange and new codes for insertions and
  deletions,'' in \emph{2019 IEEE 60th Annual Symposium on Foundations of
  Computer Science (FOCS)}, 2019, pp. 334--347.

\bibitem{Chen18}
K.~{Cheng}, Z.~{Jin}, X.~{Li}, and K.~{Wu}, ``Deterministic document exchange
  protocols, and almost optimal binary codes for edit errors,'' in \emph{2018
  IEEE 59th Annual Symposium on Foundations of Computer Science (FOCS)}, 2018,
  pp. 200--211.

\bibitem{SimaIT}
J.~{Sima} and J.~{Bruck}, ``On optimal k-deletion correcting codes,''
  \emph{IEEE Transactions on Information Theory (Early Access)}, 2020.

\bibitem{SimaSYS}
J.~{Sima}, R.~{Gabrys}, and J.~{Bruck}, ``Optimal systematic t-deletion
  correcting codes,'' in \emph{2020 IEEE International Symposium on Information
  Theory (ISIT)}, 2020, pp. 769--774.

\bibitem{Sch17}
C.~Schoeny, A.~Wachter-Zeh, R.~Gabrys, and E.~Yaakobi, ``Codes correcting a
  burst of deletions or insertions,'' \emph{IEEE Transactions on Information
  Theory}, vol.~63, no.~4, pp. 1971--1985, 2017.

\bibitem{GClocalized}
S.~Kas~Hanna and S.~El~Rouayheb, ``Codes for correcting localized deletions,''
  \emph{IEEE Transactions on Information Theory}, vol.~67, no.~4, pp.
  2206--2216, 2021.

\bibitem{A20}
A.~{Lenz} and N.~{Polyanskii}, ``Optimal codes correcting a burst of deletions
  of variable length,'' in \emph{2020 IEEE International Symposium on
  Information Theory (ISIT)}, 2020, pp. 757--762.

\bibitem{K03}
S.~{Konstantinidis}, S.~{Perron}, and L.~A. {Wilcox-O'Hearn}, ``On a simple
  method for detecting synchronization errors in coded messages,'' \emph{IEEE
  Transactions on Information Theory}, vol.~49, no.~5, pp. 1355--1363, 2003.

\bibitem{P13}
F.~{Palunčić}, K.~A.~S. {Abdel-Ghaffar}, and H.~C. {Ferreira},
  ``Insertion/deletion detecting codes and the boundary problem,'' \emph{IEEE
  Transactions on Information Theory}, vol.~59, no.~9, pp. 5935--5943, 2013.

\end{thebibliography}

\balance

\end{document}